\documentclass{CSML}
\pdfoutput=1

\usepackage{lastpage}
\newcommand{\dOi}{14(1:20)2018}
\lmcsheading{}{1--\pageref{LastPage}}{}{}%
{Jul.~02, 2015}{Mar.~06, 2018}{}

\usepackage{hyperref}
\hypersetup{hidelinks}
\usepackage{latexsym}
\usepackage{amsmath, amsthm, amssymb}
\usepackage{color}
\usepackage{graphicx,graphics}

\DeclareMathSymbol{\Not}{\mathrel}{symbols}{"36}

\DeclareMathSymbol{\not}{\mathrel}{symbols}{"36}
\def\NEQ{\Not=}

\def\nat{\mathbb{N}}
\def\real{\mathbb{R}}

 \def\vp{\varphi}
 \def\mM{\mathcal{M}}

\def\mC{\mathcal{C}}

\def\INF{\mathit{INF}}
\def\Beet{\mathit{Between}}

\def\suc{\mathit{suc}}
\def\form{\mathit{Form}}
\def\TL{\mathit{TL}}
\def\TLs{\mathit{TL}(\until,\since,\suntil,\ssince)}
\newcommand     {\KMINUS}   {\textbf{\textsf {K}} ^{-}}
\newcommand     {\KPLUS}    {\textbf{\textsf {K}} ^{+}}
\newcommand{\until}{{\sf Until}}
\newcommand{\since}{{\sf Since}}
\newcommand{\suntil}{{\sf Until}^s}
\newcommand{\ssince}{{\sf Since}^s}
\newcommand{\G}{{\square}}

\newcommand{\Gb}{\overleftarrow{{\square}}}
\newcommand{\untg}{{\mbox{Until-gap}}}

\def\bgamma{\boldsymbol{\gamma}}

\def\andl           {\wedge}                
\def\AND            {\bigwedge}
\def\orl            {\vee}                  

\def\not            {\neg}
\def\E              {\exists}
\def\A              {\forall}

\def\true              {\textbf{True}}
\def\false              {\textbf{False}}
\def\cond              {\text{Cond}}

\def\In1{\  \text{In}_{(d,c)}}
\newcommand{\EA}    { \overrightarrow{\exists}  \forall}
\def\btw{\mathit{Between}}

\newcommand{\Alub}[3] {(\A {#1})_{> {#2}}^{< {#3}}}     

\newcommand{\Elub}[3] {(\E {#1})_{> {#2}}^{< {#3}}}     

\newcommand{\Part}[2]{ \texttt{Part} (\langle  {#1} \rangle ,{#2}  )}
\newcommand{\Parts}[3]{ \texttt{Part} (\langle  \delta_{#1}, \dots, \delta_{#2}  \rangle ,{#3}  )}
\newcommand{\sh}[1]{O_{sh( #1)}}

\numberwithin{subcase}{case}

\newcommand{\emp}[1]    {\textbf {\textit{#1}}}

\def\FOMLO      {First-Order Monadic Logic of Order}
\def\FOMLOb     {First-Order Monadic Logic of Order }

\def\FOMLOi     {\textit{FOMLO}}
\def\FOMLOiB    {\textit{FOMLO} }

\newcommand{\IFOMLO}{\mathit{FOMLO}}

\newcommand{\mTL}{\mathit{TL}}
\newcommand{\TLany}             {\mTL(B)}
\def\mI{\mathcal{I}}
\newcommand{\g}[1]      {\mathcal{#1}}
\newcommand     {\mdlty}    {{\sf M}}
\newcommand     {\mdltyB}   {{~\sf M ~}}

\newtheorem{xmpl}[thm]{Example}
\def\mE{\mathcal{E}}
\begin{document}
\title { A  Proof  of  Stavi's  Theorem}

\author{Alexander   Rabinovich}
  \address{ The Blavatnik  School    of Computer Science, Tel Aviv
  University,
          Tel Aviv 69978, Israel}
        \email{rabinoa@post.tau.ac.il}


\begin{abstract}
   Kamp's theorem established  the expressive equivalence of
the  temporal logic  with   Until and Since  and  the First-Order
Monadic Logic of Order (FOMLO) over the Dedekind-complete  time
flows. However, this temporal logic  is not expressively complete
for
 $\IFOMLO$ over the rationals.
 Stavi introduced two additional
         modalities
        and proved that the temporal logic with Until, Since  and  Stavi's  modalities
    is expressively equivalent to  $\IFOMLO$ over all linear orders.
      We present   a simple proof of Stavi's theorem.
\end{abstract}
 \maketitle

\section{Introduction}\label{sect:intro}


Temporal Logic ($\TL$) introduced  to Computer Science by Pnueli in
\cite{Pnu77}
 is a convenient framework for 
reasoning about ``reactive'' systems.  This has  made temporal
logics  a popular subject in the Computer Science community,
enjoying extensive research. 
In $\TL$ we describe basic system properties by {\it atomic
propositions\/} that hold at some points in time, but not at others.
More complex properties are conveyed by formulas built from the
atoms using Boolean connectives and {\it Modalities\/} (temporal
connectives):
A $k$-place modality $M$ transforms statements $\varphi_1,\dots,
\varphi_k$ possibly on `past' or `future' points to a statement
$M(\varphi_1,\dots,\varphi_k)$ on the `present' point $t_0$.
The rule to determine the truth of a statement
$M(\varphi_1,\dots,\varphi_k)$ at $t_0$ is called a {\it Truth
Table\/}. The choice of particular modalities with their truth
tables yields different temporal logics. A temporal logic with
modalities   $M_1,\dots, M_k$ is denoted by
 $\TL(M_1,\dots ,M_k)$.

The simplest example is the one place modality $\Diamond   P$
saying:
``$P$  holds some time in the future."  Its truth table is 
formalized by $\varphi_{{}_{\!\Diamond }} (t_0,P)\equiv (\exists
t>t_0) P(t)$. This is a formula  of the First-Order  Monadic Logic
of
Order (\FOMLOi) - 
a fundamental formalism in Mathematical Logic
  where formulas are built using atomic propositions $P(t)$,
   atomic relations between elements $t_1=t_2$, $t_1<t_2$, Boolean
connectives and  first-order quantifiers $\exists t$  and $\forall
t$. Most modalities used in the literature are defined by such
\FOMLOi \
truth tables, and as a result, every temporal formula 
translates directly into an equivalent  $\IFOMLO$ formula. Thus,
different temporal logics may be considered as a convenient way to
use fragments of $\IFOMLO$.  $\IFOMLO$  can also
serve as a yardstick by which  one is able to check the strength of temporal logics: 
A temporal logic is {\it expressively complete\/} for a fragment $L$
of $\IFOMLO$ if every formula of $L$ with a single free variable
$t_0$ is equivalent to a temporal formula.

Actually, the notion of expressive completeness   refers to a
temporal logic and to a model (or a class of models)
 since the
question whether two formulas are equivalent depends on the domain
over which they are evaluated.  Any  (partially) ordered set with
monadic predicates is a model for $\TL$ and \FOMLOi, but the main,
{\it canonical\/}, linear time intended models  are the non-negative
integers $\langle \nat,~<\rangle$  for discrete time and the
non-negative reals $\langle \real^{\geq 0},<\rangle$  for continuous time.

A major result concerning $\TL$ is Kamp's theorem \cite{Kam68},
which implies that the pair of modalities ``$\mbox{$P{~\it Until}
~Q$}$"\ { and}\ ``$\mbox{$P{~\it Since} ~Q$}$" is expressively
complete for $\IFOMLO$ over the above two linear time  canonical
models.

The temporal logic with the modalities \emph{Until} and \emph{Since}
is not expressively complete for
 $\IFOMLO$ over the rationals \cite{GHR94}.

 Stavi introduced two additional
         modalities $ \suntil$ and $ \ssince $ (see Sect. \ref{subsect:TL})
        and proved that $\TLs$  is expressively
complete for $\IFOMLO$ over all linear orders. There are only two
published  proofs of Stavi's theorem \cite{GHR93,GHR94}; however,
none    is simple.

The objective of this paper is to present  a simple proof  of
Stavi's theorem.

The rest of the paper is organized as follows: In Sect.
\ref{sect:prel} we recall the definitions of the monadic logic, the
temporal logics and state  Kamp's and Stavi's  theorems. In  Sect.
\ref{sect:part} we introduce  partition formulas which play an
important role in our proof of Stavi's theorem.
 In
Sect. \ref{sect:proof-stavi}  we prove  Stavi's theorem. The proof
of one  proposition is  postponed to Sect.
 \ref{sect:s-comp}. Sect.
\ref{sect:related}  comments  on  the   previous  proofs of Stavi's
theorem.

\section{Preliminaries}\label{sect:prel}
In this section
 we recall the definitions of   linear orders,  the first-order  monadic logic of order,
the temporal logics and state  Kamp's and Stavi's  theorems.

\subsection{Intervals and gaps in linear orders}

A subset $I$ of a linear order $( {T}, <)$ is an \emph{interval}, if
for all $t_1 < t < t_2$ with $t_1, t_2 \in I$ also $t \in I$.
 For intervals with
endpoints $a, b \in {T}$, whether open or closed on either end, we
will use the standard notation, such as $[a, b): = \{ t \in {T} \mid
a \leq t < b \}$,  $(a, b): = \{ t \in {T} \mid a < t < b \}$, etc.
For   $a\in {T}$  let $[a,\infty) := \{t \mid  t \geq  a \}$ and,
similarly,   $(-\infty, a)  := \{ t\in
 \mid t<a\}$.

  A  \emph{Dedekind cut} of a linearly ordered set
$( {T}, <)$
 is a downward closed non-empty set $C \subseteq {T}$ such that  its complement is non-empty and if $C$
has a least upper bound in $( {T}, <)$,  then it is contained in
$C$. A \emph{proper cut}  or a \emph{gap} is a cut that has no least
upper bound in $( {T}, <)$,  i.e., one that has no maximal element.

A linear order is \emph{Dedekind complete} if it has no gaps;
equivalently, if  for every non-empty subset $S$ of $T$, if $S$ has
a lower bound in $T$, then it has a greatest lower bound, written
$\inf(S)$, and if $S$ has an upper bound in $T$, then it has a least
upper bound, written $\sup(S)$.

For a gap $g$  and an element $t\in T$ we write $t<g$ (respectively,
$g<t$) if $t\in g$ (respectively, $t\notin g$). We also write
$(t,g)$ for the interval $\{a\in T\mid a>t\wedge a\in g\}$;
similarly, $(g,t)$  is  $\{a\in T\mid a<t\wedge a\notin g\}$.
Finally, for gaps $g_1$ and $g_2$ we write $g_1\leq g_2$ if
$g_1\subseteq g_2$, and the interval $(g_1,g_2)$ is defined as
$\{a\in T\mid a\notin g_1 \wedge a\in g_2\}$.

\subsection{First-order Monadic Logic and Temporal Logics} \label{sect:prelim}

We present the basic definitions of \FOMLOb (\FOMLOi) and {Temporal
Logic } ($\TL$), and   well-known results concerning their
expressive power. Fix a set $\Sigma$ of \emph{atoms}. We use $P, Q,
R,   \dots$ to denote members of $\Sigma$. The syntax and semantics
of both logics are defined below with respect to such a $\Sigma$.


\subsubsection{\FOMLO} \label{subsect:fomlo} \hfill



\emp{Syntax:} In the context of \FOMLOi~  the atoms of $\Sigma$ are
referred to (and used) as \emp{unary predicate symbols}. Formulas
are built using these symbols,  plus two binary relation symbols:
$<$
and $=$, and a  set of {first-order  variables} 
(denoted: $ x,y, z, \dots $).
Formulas are defined by the grammar:
$$atomic ::= ~~x < y ~~|~~ x = y ~~|~~P(x)~~~~~~~~~~ (\mbox{where} ~~ P \in \Sigma)$$
$$\varphi ::= ~~ atomic ~~|~~ \neg \varphi_1 ~~|~~ \varphi_1 \orl \varphi_2 ~~|~~ \varphi_1 \andl \varphi_2 ~~|~~ \exists x \varphi_1 ~~|~~ \forall x \varphi_1$$
The notation $\varphi(x_1, \dots, x_n)$ implies that $\varphi$ is a 
formula where the $x_i$  are the only variables occurring free;
writing $\varphi(x_1, \dots, x_n, P_1, \dots, P_k)$ additionally
implies that the $P_i$ are the only predicate symbols that occur in
$\varphi$. We will also use the standard abbreviated notation for
\emp{bounded quantifiers}, e.g.: $(\exists x)_{> z}(\dots)$ denotes
$\exists x ((x > z) \andl (\dots))$ and $(\forall x)^{<
z}_{>z_1}(\dots)$ denotes $\forall x ((z_1<x < z) \rightarrow
(\dots))$, etc.

\emp{Semantics}: 
Formulas are interpreted over \textit{labeled  linear orders} which are  called \emph{chains}. A $\Sigma$-\emp{chain}
 is a triplet $\mM = ( {T}, <, \mI)$ where
${T}$ is a set - the \emp{domain} of the chain, 
$<$ is a linear  order relation on ${T}$, 
and $\g{I} : {\Sigma} \rightarrow \g{P} ({T})$ is the
\emp{interpretation} of  $\Sigma$ (where $\g{P}$ is the powerset
notation). We use the standard notation $\mM, t_1, t_2, \dots t_n
\models \varphi(x_1, x_2, \dots x_n)$ to indicate that the formula $\varphi$ with free variables among $x_1,\dots,x_n$
is satisfiable in $\mM$ when $x_i$ are interpreted as elements $t_i$ of $\mM$.
For atomic $P(x)$ this is defined by: $\mM, t \models P(x)$ iff $t
\in \g{I}(P)$; The semantics of $<, =, \neg, \andl, \orl, \exists
\mbox { and } \forall$ is defined in a standard way.


\subsubsection{ Temporal Logics } \label{subsect:TL}\hfill


\emp{Syntax:} In the context of $\TL$  the atoms of $\Sigma$ are
used as \emph{atomic propositions} (also called \emph{propositional
atoms}). Formulas are built using these atoms  and a set (finite or
infinite) $B$ of
\emph{modality names}, 
where an integer \emph{arity}, denoted $|\mdlty|$, is associated with
each $\mdlty \in B$.
The syntax of  $\TL$ with the \emp{basis} 
$B$, denoted \emp{TL(B)}, 
is defined by the grammar:
$$F ::= ~~P ~~|~~ \neg F_1 ~~|~~ F_1 \orl F_2 ~~|~~ F_1 \andl F_2 ~~|~~ \mdlty(F_1,F_2,\dots,F_n),$$
where $P \in \Sigma$ 
and $\mdlty \in B$ an $n$-place modality 
(with arity $|\mdlty| = n$). As usual, \emp{True} denotes $P \orl
\neg P$ and \emp{False} denotes $P \andl \neg P$;
we will use infix notation for binary modalities, where $F_1 \mdltyB
F_2$ is an alternative notation for $\mdlty(F_1, F_2)$.

\emp{Semantics:} 
Formulas are interpreted at \emph{time-points} (or \emph{moments}) in
chains  (elements of the domain).
The domain ${T}$ of $\mM = ( {T}, <, \mI)$ is called the
\emph{time domain}, and $({T}, <)$ - the \emph{time flow} of the
chain.
The semantics of each $n$-place 
modality $\mdlty \in B$ is defined by
                a `rule' specifying how the set of moments
                where $\mdlty(F_1, \dots, F_n)$ holds (in a given structure) is determined by the $n$ sets of moments where each of
                the formulas $F_i$ holds. Such a `rule' for $\mdlty$ is formally specified by
                an operator $\g{O}_{\mdlty}$ on time flows, where
                given a time flow $\g{F} = ({T}, <)$,
                $\g{O}_{\mdlty} (\g{F}) $ is  an operator in $(\g{P}({T}))^n \longrightarrow \g{P}({T})$.

The semantics of $\TLany$ formulas is then defined inductively.
Given a chain  $\mM = ( {T}, <, \mI)$ and a moment $t \in \mM$,
 define when a 
 formula $F$ \emph{holds} in $\mM$ at $t$ - denoted $\mM, t \models F$:

\begin {itemize}
    \item   $\mM, t \models P$ iff $t \in \g{I}(P)$, for any propositional atom $P$.
    \item   $\mM, t \models F \orl G$ iff $\mM, t \models F$ or $\mM, t \models G$;
                similarly  for $\andl$ and $\neg$.
    \item $\mM, t \models \mdlty(F_1, \dots, F_n) \mbox { iff } t \in [\g{O}_{\mdlty} ({T}, <)](T_1, \dots, T_n)$
                where $\mdlty \in B$ is an $n$-place modality, $F_1, \dots, F_n$ are formulas and
                $T_i =_{def} \{ s \in {T} : \mM, s \models F_i \}$.
\end {itemize}


\emp{Truth tables:}
Practically most standard modalities studied in the literature can
be specified in \FOMLOi:
    A \FOMLOiB formula
                $\varphi(x, P_1, \dots, P_n)$ (with a single free first-order  variable $x$ and with $n$ predicate symbols $P_i$) is called
                an \emp{$n$-place  first-order  truth table}. 
                Such a truth table $\varphi$ {defines}
                an $n$-ary modality $\mdlty$ (whose semantics is given by an operator $\g{O}_{\mdlty}$) iff
                for any time flow $({T}, <)$, for any 
                $T_1, \dots, T_n \subseteq {T}$ 
                and for any structure $\mM = ({T}, <, \g{I})$ where $\g{I}(P_i) = T_i$:
                $$[\g{O}_M({T}, <)](T_1, \dots, T_n) = \{ t \in {T} : \mM, t \models \varphi(x, P_1, \dots, P_n) \}$$

\begin{xmpl} \label{ex:modalities}
Below are truth-table definitions for the (binary)
\emp{strict}-$\until$ and \emp{strict}-$\since$  and
the (unary) $\G $,  $\Gb$,   $\KPLUS$ and $\KMINUS$:

\begin {itemize}

%

        \item   $P~ \until~ Q$ is defined by
        : $\varphi_{_{\until}}(x, P, Q) := (\exists x') _{> x} (Q(x') \andl (\forall y) _{> x} ^{< x'}
        P(y))$.

        \item   $P~\since~Q $ is defined by:
                    $\varphi_{_{\since}}(x, P, Q): = (\exists x') ^{< x} (Q(x') \andl (\forall y) _{> x'} ^{< x}
                    P(y))$.
  \item $\G (P)$  (respectively, $\Gb(P) $) - ``$P$ holds everywhere after (respectively, before)  the
current moment": \begin{gather*}\varphi_{_\G}(x, P) :=(\forall x')_{> x} P(x')\\
\varphi_{_{\Gb}}(x, P) :=(\forall x')_{< x} P(x')
\end{gather*}
%
        \item   $\KPLUS$ defined by:
                    $\varphi_{_{\KPLUS}}(x, P) := (\forall x') _{> x} (\exists y) _{> x} ^{< x'}
                    P(y))$.

        \item   $\KMINUS$ defined by:
                    $\varphi_{_{\KMINUS}}(x, P) := (\forall x') ^{< x} (\exists y) _{> x'} ^{< x}
                    P(y))$.

\end {itemize}
\end{xmpl}
Formula $\KMINUS(P)$   holds at a moment $t $ iff $t=\sup(\{t' \mid
t'<t\wedge  P(t')\})$. Dually,    $\KPLUS(P)$  holds at $t$ iff
$t=\inf(\{t' \mid t'>t\wedge  P(t')\})$. Note that $\KPLUS(P)$ is
equivalent to
 $\neg((\neg P) \until \textbf{True})$  and   $\G  P$
   is  equivalent to $\neg( \true\until (\neg P))$.


Let $\bgamma^+$ be a unary modality such that  $\bgamma^+ (P)$ holds
at $t$ if there is a gap $g>t$ in the order such that
$g= \sup(\{t' \mid \Alub{y}{t}{t'} P(y)\})$.
 We say that $g$ is the gap
left  definable by $P$ that succeeds $t$, or just that \emph{$g$ is
$P$-gap that succeeds $t$};
$P$ holds
everywhere on the interval $(t,g)$, and for every  $t_1>g$,  there
is $t'\in(g,t_1)$ such that $t'\notin P$.

%
A natural formalization of  $\bgamma^+$  semantics  uses a
second-order quantifier - ``there is a gap''; however,
$\bgamma^+(P)$ is equivalent to the conjunction of the following
  formulas \cite{GHR94}:
 \begin{enumerate}
\item $(P\until P) \wedge \neg (P\until \neg P)$.
\item $  \neg \G  P$ - ``$\neg P$ holds somewhere in the future."
\item $\neg (P\until  (P\wedge  \KPLUS (\neg P)))$.
%
\end{enumerate}
Since  $\G$ and $\KPLUS$ are equivalent to $\TL(\until)$ formulas,
$\bgamma^+(P)$ can be considered as an abbreviation of a
$\TL(\until)$ formula, and   $\bgamma^+$ has a first-order truth
table
 $\varphi_{\bgamma^+}(x, P)$.

$\bgamma^-$ is the mirror image of $\bgamma^+$, i.e., going into the
past instead of into the future.
 $\bgamma^- (P)$ holds
at $t$ if there is a gap $g<t$ in the order such that $P$ holds
everywhere on the interval $(g,t)$, and for every  $t_1<g$,  there
is $t'\in(t_1,g)$ such that $t'\notin P$. We say that $g$ is the
(right  definable) $P$-gap that precedes $t$, or just that \emph{$g$ is
$P$-gap that  precedes  $t$.}

The modalities $\suntil$ and $\ssince$ were introduced by Stavi.
$P\suntil Q$ holds at $t$ if there is a gap $g>t$ such that:
\begin{itemize}
\item $P$ is true on $(t,g)$.
\item In the future of the gap, $P$ is false arbitrarily close to the
gap, and
\item $Q$ is true from $g$  into the future for some uninterrupted
stretch of time.
\end{itemize}

 $\suntil$ has a first-order truth table
$\varphi_{\suntil}(x, P,Q)$ which is the  conjunction of the
following
  formulas: \begin{enumerate}

\item $\varphi_{\bgamma^+}(x, P)$.

\item $(\exists x_1)_{> x} \big(\neg P(x_1)
\wedge     (\forall y) _{> x } ^{< x_1}[ (\neg P(y))\rightarrow
\Alub{z}{y}{x_1} Q(z) ]\big)$.

\end{enumerate}
$\ssince$ is the mirror image of $\suntil$.


\subsection{Kamp's and Stavi's  Theorems} \label{subsect:exComp}


We are interested in the relative expressive power of $\TL$
(compared to \FOMLOi) over the class of \emp{linear structures},
where the time flow is an irreflexive linear order.

\emp{Equivalence} 
between temporal and monadic formulas is naturally defined: $F$ is
equivalent to $\varphi(x)$  over a class $\mC$ of structures iff for
any $\mM\in \mC$   and $t \in \mM$: $\mM, t \models F
\Leftrightarrow \mM, t \models \varphi(x)$. If $\mC$ is the  class of
all chains, we will say that $F$ is equivalent to $\varphi$.


\emp{Expressive completeness/equivalence}: A temporal language
$\TLany$ 
 is \emph{expressively complete } for
 \FOMLOiB over a class $\g{C}$ of structures iff for every
$\FOMLOiB $ formula  $\varphi(z)$    with  one free variable  there is a $\psi\in \TLany $ such that $\varphi$ is equivalent to $\psi$ over $\g{C}$.
 Similarly, one may speak of {expressive
completeness } of \FOMLOiB
for some temporal language. 
If we have {expressive completeness } in both directions between two
languages, then they are {expressively equivalent}.

If every modality in $B$ has a \FOMLOiB truth-table, then it is easy to translate every formula of  $\TLany$  to an equivalent \FOMLOiB  formula.
Hence, in this case  \FOMLOiB is expressively  complete for $\TLany$.

%
The fundamental theorem  of Kamp's 
states:

\begin{thm}[\cite{Kam68}] \label{th:kamp}
     $\TL(\until,\since)$ is expressively equivalent to \FOMLOiB over Dedekind complete chains.
\end{thm}


 $\TL(\until,\since)$  is not expressively complete for
 $\IFOMLO$ over the rationals \cite{GHR94}.
 Stavi  introduced two new
         modalities $ \suntil$ and $ \ssince $ (see Sect. \ref{subsect:TL})
        and proved:
        \begin{thm}
          $\TL( \until, \since, \suntil, \ssince )$  is expressively equivalent to \FOMLOiB over all  chains.
               \end{thm}
As $\until,~\since$ and  Stavi's  modalities are definable in
\FOMLOi, it follows that \FOMLOiB is {expressively complete } for
$\TLs$. The contribution of our paper is  a proof  that $\TLs$ is
{expressively complete } for \FOMLOi.
 Our proof is constructive. An algorithm which for
every $\IFOMLO$ formula $\varphi(x )$ constructs a $\TLs$ formula
which is equivalent to $\varphi$
 is  easily extracted from our proof.



\section{Partition Formulas}\label{sect:part}
In this section we introduce partition formulas and state their
properties. They will play an important role in our proof of Stavi's
theorem.
 The basic partition formulas generalize   the
Decomposition formulas of \cite{GPSS80}.

\begin{defi}[Partition expressions]  Let $\Sigma$ be a set of monadic predicate
names, and $\delta_1(x), \dots, \delta_n(x)$  are quantifier free
first-order formulas over $\Sigma$ with one free variable, and
$O\subseteq \{1,\dots,n\}$.  An expression  $\Part{\delta_1, \dots,
\delta_n}{O}$ is called a partition expression over $\Sigma$.
\end{defi}
\paragraph{{\textbf{Semantics}.}}
Let $I$ be an interval of a $\Sigma$-chain $\mM$. A partition
expression $\Part{\delta_1, \dots, \delta_n}{O}$ holds on $I$ in
$\mM$ (notation $\mM, I\models \Part{ \delta_1, \dots ,
\delta_n}{O}$
   if
$I$ can be partitioned into $n$ non-empty intervals $I_1,\dots ,I_n$
such that $\delta_j$ holds on all points in $I_j$, and $I_i$
precedes $ I_j$ for $i<j$, and $I_j$   a one-point interval for
$j\in O$. Note that we do not require that if $I_j$ is a one-point
interval, then $j\in O$. Observe that the semantics of partition
expressions does not depend on the names of the  variables  that
appear in $\delta_i$.

For example,  $\Part{P_1(x), P_2(x)\vee P_3(x)}{\{1,2\}}$ holds over
$I$ iff $I$ is a two point interval and $P_1 $ holds over its first
point and $P_2 $ or $P_3$ holds over its second point.
$\Part{\true, \true}{\{1\}}$ holds over $I$ iff $I$ has a minimal
point and at least two points.

\begin{defi}[Partition Formulas]  Let $\Sigma$ be a set of monadic predicate names.
\begin{description}
\item[Basic Partition Formulas]
A basic partition formula (over $\Sigma$) is an expression of one of
the following forms:
\begin{enumerate}
\item    $z=y$ or $z<y$
\item
 $\Parts{1}{n}{O} [y,z]$

\item
$\Parts{1}{n}{O} [z,\infty)$ or $\Parts{1}{n}{O} (-\infty,z]$ or \\
$\Parts{1}{n}{O} (-\infty,\infty)$,
\end{enumerate}
where $\Parts{1}{n}{O} $ are partition  expressions.
\item[Partition Formulas] are constructed from the  basic partition
formulas     by Boolean connectives and existential quantifier.
\item[Simple Partition Formulas] are  constructed  from the  basic partition
formulas     by conjunction and disjunction.

\item[Normal Partition Formulas]
    A Normal Partition Formula  
      is a partition formula   of the form:
\begin{align*} \displaystyle
E(z_1,\dots, z_m)   :=
   &
                          \left(\bigwedge_{k=n+1}^m z_k=z_{i_k}
                        \right) \wedge
                          \left(z_1 < z_2<  \dots <z_n \right) \\
                         &
                        \andl~\AND _{j=2}^{n} W_j  [z_{j-1}, z_{j} ]  \\
    & \andl~   W_{n+1} [z_n,\infty)     \wedge
                         W_1( -\infty,z_1]   \\
                          & \andl~   W_0  (-\infty,\infty)                                                                                                                                                       \\
\end{align*}
where $W_j$ are basic partition formulas, $n\leq m$  and
$i_{n+1},\dots, i_m\in\{1,\dots, n\}$.
\end{description}
\end{defi}
The  semantics of the partition formulas  will  not depend on the
names of   variables  that occur  in    partition expressions. These
occurrences of the variables are considered to be bound. For other
occurrences of variables the definition whether occurrences are free
or bound is standard.

\paragraph{{\textbf{Semantics}.}}  Partition formulas are interpreted over $\Sigma$-{chains}.
Let   $\mM = ( {T}, <, \mI)$ be a $\Sigma$-{chain}.
 We use the standard notation $\mM, t_1, t_2, \dots t_n
\models \varphi(x_1, x_2, \dots x_n)$ to indicate that the formula
$\varphi$ with free variables among $x_1,\dots,x_n$ is satisfiable
in $\mM$ when $x_i$ are interpreted as elements $t_i$ of $\mM$. For
basic partition formulas  this is defined by: $\mM, t \models
 \Parts{1}{n}{O} [x,\infty)$    iff  the partition expression
$\Parts{1}{n}{O}$ holds on the interval $[t,\infty)$ in $\mM$;
similarly, $\mM, t \models
 \Parts{1}{n}{O} (-\infty, x]$  (respectively, $\mM  \models
 \Parts{1}{n}{O} (-\infty ,\infty))$ iff $\Parts{1}{n}{O}$ holds on the interval $(-\infty,t] $
 (respectively, the interval  $ (-\infty ,\infty)$)  in
 $\mM$; and $\mM, t_1,t_2 \models
 \Parts{1}{n}{O} [x_1,x_2 ]$ iff $\Parts{1}{n}{O}$ holds on the interval $[t_1,t_2]
 $ in $\mM$; the semantics of $<, =, \neg, \andl, \orl,
 \mbox { and } \exists$ is defined in a standard way.

The following lemmas immediately follow  from the definitions and standard logical equivalences.
\begin{lem}\label{lem:triv1}
\begin{enumerate}

\item  Every simple formula is equivalent to a disjunction of normal formulas.

\item For every normal  formula $\vp$, 
 the formula  $\exists x \vp$ is equivalent to a disjunction of  normal formulas.
\item For every simple formula $\vp$,
 the formula  $\exists x \vp$ is equivalent to a simple formula.
\end{enumerate}
\end{lem}

 \begin{lem}[Closure properties]\label{lem:closure1}
The set of simple formulas is (semantically)  closed under
disjunction, conjunction, and existential quantifier.
\end{lem}
The set of simple  formulas is not closed under negation.
However, we show   later (see Proposition \ref{lem:sneg}) that the
negation of a  simple   formula is equivalent to a simple  formula
in the expansion of the  chains by all
 $\TLs$ definable predicates.

In the rest of this  section we explain  how to translate   a simple
partition formula
  with one free variable  into an  equivalent
 $\TLs$ formula.

Let $\delta_1, \dots , \delta_k$ be quantifier free first-order
formulas with
 one free variable and  $O\subseteq \{1,\dots ,k\}$. For $i=1,\dots k$,
let $D_i$ be a  temporal formula equivalent to $\delta_i$. Define:
\begin{equation}\label{eq:tr-1}
F_k:=D_k
\end{equation}
\begin{equation} \label{eq:tr-2}
 F_{i-1} := D_{i-1}\wedge
 \left\{ \begin{array}{ll}
         \false \until F_i  & \mbox{if $i-1\in O$ and $i\in O$};\\
        D_{i} \until F_i   & \mbox{if $i-1\in O$ and $i\notin O$};\\
        D_{i-1} \until F_i   & \mbox{if $i-1\notin O$ and $i\in O$};\\
D_{i-1} \until^* F_i   & \mbox{if $i-1\notin O$ and $i\notin O$},\\
        \end{array} \right.
        \end{equation}
where $ P\until^* Q$  holds at $t$ if there is $ t'>t$  such that
$\Part{P(x), Q(x)}{\emptyset}  $ holds on the interval
  $[t,t'] $; $ P\until^* Q$ can be expressed as
  disjunction of the following formulas:
  \begin{itemize}
  \item $P\wedge \big((P \until Q)\vee (Q\until Q)\vee  P\until \big (P\wedge  (Q\until Q) \big) \big)$
\item $P\wedge (P \suntil Q)$
\end{itemize}
\begin{lem} \label{lem:def-fi}
\leavevmode
\begin{enumerate}
\item Assume that there is  $t $   and a   partition of $[t_1,t]$ into non-empty
intervals  $I_1,\dots, I_k$ such that $\delta_j$ holds on $I_j$ and
$I_i$ precedes $ I_j$ for $i<j$, and $I_i$ is a one-point interval
for $i\in O$.
  Then $F_{k-j}$ holds on $I_{k-j}$.

\item if $F_{k-j}$ holds at $t_{k-j}$ then there is  $t \geq t_{k-j}$ such
that  $ \Part{\delta_{k-j}, \dots , \delta_k}{O_{k-j}}$ holds on
$[t_{k-j} ,t ]$, where $l\in O_{k-j}$ iff $l+k-1-j\in O$.
\item
$F_1$ holds at $t_1$ iff there is  $t \geq t_1$ such that $
\Part{\delta_1, \dots , \delta_k}{O}$ holds on  $[t_1,t ]$.
\end{enumerate}
\end{lem}
\begin{proof}
(1) and (2)  by induction on $j$.
(3) immediately from (1) and (2).
\end{proof}
Let $\delta'_k$ be a quantifier-free first-order formula  with
 one free variable and $D'_k$ be a temporal formula equivalent to
$\delta'_k$.
 If $k\notin   O$ and we set $D_k:=D'_k\wedge \G
D'_k$ in equation (\ref{eq:tr-1}), then $F_1$ holds at $t_1$  iff
$\Part{\delta_1, \dots ,\delta_{k-1} , \delta'_k}{O}$ holds on
$[t_1,\infty)$; if $k \in  O$ and we set $D_k:=D'_k\wedge \G \false
$ in equation (\ref{eq:tr-1}), then $F_1$ holds at $t_1$  iff
$\Part{\delta_1, \dots ,\delta_{k-1}  , \delta'_k}{O}$ holds on
$[t_1,\infty)$.
 Hence, we obtained:
\begin{lem} \label{lem:trans-infty}  For every $ \delta_1, \dots , \delta_k$
and $O\subseteq \{1,\dots,k\}$ there is a $\TL(\until,\suntil)$
formula $F$ such that $F$ holds at $t$ iff $\Part{\delta_1, \dots ,
\delta_k}{O}$  holds on $[t,\infty)$.
\end{lem}
By Lemma \ref{lem:trans-infty} and  standard logical equivalences we
obtain:
\begin{prop}[From simple formulas to $TL$]
 \label{prop:forms} 
Every simple formula     with at most one free  variable  is
equivalent
  to a  $\TLs$ formula.
\end{prop}
\begin{proof} 
Note that every simple partition formula with at most one free
variable $z$ is equivalent to a boolean combination of basic
partition formulas of the forms:   $ \texttt{Part} (\langle \delta_1
\rangle ,O )[z,z]  $, $\Parts{1}{n}{O} [z,\infty)$, $\Parts{1}{n}{O}
(-\infty,z]$, or $\Parts{1}{n}{O} (-\infty,\infty)$. Let $D_1$ be a
temporal formula equivalent to the first-order quantifier free
formula $\delta_1$. A formula of the form $ \texttt{ Part} (\langle
\delta_1  \rangle ,O )[z,z]  $ is equivalent to $D_1$.
 By Lemma
\ref{lem:trans-infty} and its mirror variant the formulas of the
 second and the third forms are equivalent to $\TLs$
formulas. \sloppy{ A formula of the form $\Parts{1}{n}{O}
(-\infty,\infty)$ is equivalent to ``$ \texttt{Part} (\langle
\delta_1  \rangle ,O )(-\infty,z]  \wedge  \Parts{1}{n}{O}
[z,\infty)$   for some $z$." }Since each of the conjuncts is
equivalent to a temporal formula, the conjunction is also equivalent
to a temporal formula $A$, and
  $\Parts{1}{n}{O}
(-\infty,\infty)$ is equivalent to $\neg(  \G \neg A  \wedge \neg
A\wedge \Gb \neg A )$.
%
 Hence, every simple formula with at most one free variable is
equivalent to a $\TLs$ formula.
\end{proof}

\section{Proof of Stavi's Theorem}\label{sect:proof-stavi}
The next definition plays a major role in our proof of Stavi's
theorem;  a similar definition is used in the proof of Kamp's
theorem \cite{GPSS80}.
\begin{defi} Let $\mM$ be a $\Sigma$ chain. We denote by  $\mE[\Sigma]$  the set of unary predicate names
$\Sigma\cup\{A~\mid A$ is an $\TLs$-formula over $\Sigma ~\}$.
 The canonical $\TLs $-expansion of $\mM$ is  an expansion of
$\mM$ to an $\mE[\Sigma]$-chain,  where each predicate name $A\in
\mE[\Sigma]$ is interpreted as $\{a\in \mM\mid \mM,a\models
A\}$\footnote{ We often  use  ``$a\in \mM$'' instead of ``$a$ is an
element of the domain of $\mM$.''}.
\end{defi}

Note that if $A$ is a $\TLs $ formula  over $\mE[\Sigma]$
predicates, then it is equivalent to a $\TLs  $ formula over
$\Sigma$, and hence to an atomic formula in the canonical $\TLs
$-expansions.

 From now on
 we say  ``formulas are equivalent  in a
 chain $\mM$''   instead of ``formulas are equivalent  in the
canonical $\TLs $-expansion of $\mM$.'' The partition formulas   are
defined as previously, but now they can use as atoms $\TLs
 $ definable predicates.

It is clear that   the results stated  in Sect. \ref{sect:part} hold
for this modified notion of partition formulas. In particular, every
simple formula     with at most one free  variable  is equivalent
  to a  $\TLs$ formula, and
 the set of simple formulas is
closed under conjunction, disjunction and existential
quantification. However, now the set of simple formulas is also
closed under negation, due to
 the next proposition whose proof is postponed to Sect.
\ref{sect:s-comp}.

\begin{prop}[Closure under Negation] \label{lem:sneg}
 The negation of every simple  partition formula  
is equivalent to    a simple partition formula.

\end{prop}
As a consequence we obtain:
\begin{prop}\label{prop:fo2ea}
 Every first-order formula
is equivalent to  a simple formula.

\end{prop}
\begin{proof} We proceed by   structural induction.
\begin{description}
\item[Atomic]
It is clear that every atomic formula is equivalent to a simple
 formula.


\item[Negation]
By  Proposition \ref{lem:sneg}.


\item[$\exists$-quantifier and disjunction ]
This  follows from Lemma \ref{lem:closure1}.  \qedhere
\end{description}
\end{proof}

Proposition \ref{prop:fo2ea} and Proposition \ref{prop:forms}
immediately imply
  Stavi's   Theorem:
\begin{thm}\sloppy{
Every  $\FOMLOiB$ formula   with one free variable
     is equivalent to  a $ \TLs  $  formula.}
\end{thm}
 This completes our proof of Stavi's theorem except 
for Proposition \ref{lem:sneg}
  which is proved in Sect.
\ref{sect:s-comp}.

\section{Proof of Proposition \ref{lem:sneg}} \label{sect:s-comp}
Throughout our  proof we will  freely use that the following
assertions and their negations are expressible by simple formulas:
\begin{enumerate}
\item $(z_0,z_1)$ contains a point in $P$.
\item   $\suc(z_0,z_1)$  -  $z_1$ is a successor of $z_0$.
\item interval $(z_0,z_1)$ contains  exactly $k$ points.
\item interval $(z_0,z_1)$ contains  at most $k$ points.
\end{enumerate}
Let us introduce some helpful notations.
\begin{nota} \sloppy{We use the  abbreviated notations  $\Parts{1}{n}{O} (z_0,z_1)$
for   $ \Part{\true,\delta_1, \dots , \delta_n,\true}{O'}[z_0,z_1]$,
where $O':=\{1,n+2\}\cup \{i+1\mid i\in O\}$. Hence,
$\mM,t_0,t_1\models \Parts{1}{n}{O} (z_0,z_1)$ iff
$\Parts{1}{n}{O}$ holds on the open interval $(t_0,t_1)$ in $\mM$.
Similarly, $\Parts{1}{n}{O} (z_0,z_1]$  stands  for   $
\Part{\true,\delta_1, \dots , \delta_n}{O'}[z_0,z_1]$, where
$O':=\{1\}\cup \{i+1\mid i\in O\}$; and $\Parts{1}{n}{O} [z_0,z_1)$
for   $ \Part{\delta_1, \dots , \delta_n,\true}{O'} [z_0,z_1]$,
where  $O':=\{n+1\}\cup O$.}
\end{nota}
By  Proposition \ref{prop:forms} and standard logical equivalences
we obtain: 
\begin{lem}\label{lem-reduction}
\sloppy If every formula of the form $\neg \Parts{1}{n}{O}
(z_0,z_1)$ is equivalent to a simple formula,  then the negation of
every simple formula is equivalent to a simple formula.
\end{lem}
 \begin{proof}
\leavevmode
%

\begin{enumerate}
\item Every basic partition formula $\varphi$ either (a)  has at most one free  variable
and then $\varphi $ and $\neg \varphi$ are equivalent to   simple
formulas by Proposition   \ref{prop:forms},  or (b) 
is equivalent to a formula of the form $\Part{\delta_1, \dots ,
\delta_k}{O }[z_0,z_1]$.
\item
 A formula of the form $\Part{\delta_1, \dots ,
\delta_k}{O }[z_0,z_1]$ is equivalent to a formula constructed by
disjunction and conjunction from formulas of the forms: (a)
$\Parts{1}{n}{O'} (z_0,z_1)$  and (b) $\suc(z_0,z_1)$, $z_0<z_1$,
$z_0=z_1$, $\delta_1(z_0)$ and $\delta_k(z_1)$, where $\delta_i(z)$
is a quantifier-free first-order formula. Formulas of the form (b)
and their negations are equivalent to simple formulas.
\end{enumerate}
Hence, if every formula of the form $\neg \Part{\delta_1, \dots ,
\delta_k}{O }(z_0,z_1)$ is equivalent to a simple formula,
by the definition of simple formulas, (1)-(2) and De Morgan's laws
we obtain the conclusion of the Lemma.
 \end{proof}
Lemma \ref{lem-reduction} and the next proposition immediately imply
Proposition \ref{lem:sneg}.
\begin{prop}[Closure under negation] \label{lem:snegopen}
Every formula of the form \[\neg \Parts{1}{n}{O} (z_0,z_1)\]  is
equivalent to a simple formula.
\end{prop}
Sect.   \ref{sect:proof-snegopen} contains  a  proof of Proposition
\ref{lem:snegopen}. In the next subsection we provide some useful
temporal logic formalizations. A proof  of the next proposition,
which is   very similar to the proof of Proposition
\ref{lem:snegopen} is presented in Sect.
\ref{sect:proof-lem:occ-st}.

\begin{prop} \label{lem:occ-st}
  The formula
  \begin{displaymath}
  \neg \E x_1 \dots  \E x_n
                             \left(z_0<x_1 <  \dots
  < x_n<z_1 \right)     \wedge ~\bigwedge_{i=1}^{n}  P_i(x_i)
\end{displaymath}
is
 equivalent  to a simple formula.
 \end{prop}
\subsection{Some formalizations in $\TLs$} \label{sect:some-form}
%
First, observe  that there is a $\TL(\until,\suntil)$  formula that
holds at $t$ if $t$ succeeds by
  a (left definable) $P_1$-gap  and until this gap
$P_1\wedge P_2$ holds. Indeed, the required formula is
$\untg(P_1,P_2):=\bgamma^+(P_1) \wedge  \bgamma^+(P_1\wedge P_2)
\wedge \neg ((P_1\wedge P_2)\suntil P_1)$.

Let  $\delta$ and $\delta'_1, \dots , \delta'_k$ be quantifier free
first-order formulas with
 one free variable. For $i=1,\dots k$,
let $D'_i$ be a  temporal formula equivalent to $\delta'_i$ and let
$D$  be a  temporal formula equivalent to $\delta$.

If we set $D_k:=\untg(D,D'_k)$
   in equation (\ref{eq:tr-1}) (see page \pageref{eq:tr-1})
and $D_i:=D'_i\wedge D$ for $i=1, \dots, k-1$ in equation
(\ref{eq:tr-2}),
 then $F_j(t_j)$ holds iff
 there is a $\delta$-gap  $g$ that succeeds $t_j$
  such that
 $\Part{\delta'_j, \dots , \delta'_k}{O} $ holds on $[t_j,g)$.
Hence, we obtained the following Lemma:
\begin{lem} \label{lem:trans-form2} For every $k$-tuple   $\langle \delta_1, \dots , \delta_k\rangle$,
  $O\subseteq \{1,\dots,k\}$  and $\delta$ there is a $\TL(\until,\suntil)$  formula $F$ such that
$F$ holds at $t$
if and only if
 there is a $\delta$-gap  $g$ that succeeds $t$
such that   $\Part{\delta_1, \dots , \delta_k}{O} $ holds on
$[t,g)$.
\end{lem}
\begin{lem} \label{lem:st-form3}
Suppose we are given $k\geq 1$  quantifier-free   formulas
$\delta_1,\dots ,\delta_k$
  with one free variable,
a set $O\subseteq \{1,\dots,k\}$, and points $a_1,d$ with $a_1\leq
d$. Let $F_1, \dots,F_k$ be defined as in equations (\ref{eq:tr-1})
and (\ref{eq:tr-2}) on page \pageref{eq:tr-1}. Then the following
are equivalent:
\begin{enumerate}
\item

 There are points $ a_1<a_2<\cdots <a_k\leq d$ such
that $\wedge_{i=1}^k F_i(a_i)$.
\item  There is  $b\in [a_1,d]$
such that $ \Part{\delta_1, \dots , \delta_k}{O}$ holds on
$[a_1,b]$.
\end{enumerate}
\end{lem}
\begin{proof}
$\Leftarrow$ direction. Let $I_1,\dots ,I_k$ be a partition of
$[a_1,b]$ into    non-empty intervals
 such that $\delta_j$ holds on all points in $I_j$
and $I_i$ precedes $ I_j$ for $i<j$, and $I_i$ is a one-point
interval for $i\in O$. Let  us  choose any $a_i\in I_i$ for
$i=2,\dots ,k$. Then $\wedge_{i=1}^k F_i(a_i)$  holds by Lemma
\ref{lem:def-fi}(1).

$\Rightarrow$ direction. Let $F_i$ for $i=1, \dots, k$ be as in the
lemma.
 By induction on $l\leq k$ we prove that if there are points $
a_1<a_2<\cdots <a_l$ such that $\wedge_{i=1}^l F_i(a_i)$ then there
is $b \leq a_l$ such that $ \Part{\delta_1\wedge F_1, \dots ,
\delta_l\wedge F_l}{O \cap \{1,\dots, l\}}$ holds on $[a_1,b]$.

  The basis is immediate, take $b:=a_1$.

Inductive step: $l\mapsto l+1$.

By the inductive assumption there is $b'\leq a_l$ and a partition of
$[a_1,b']$ into $l$ non-empty intervals $I'_1, \dots , I'_l$ such
that $\delta_i\wedge F_i$ holds on $I'_i$ for $i\leq l$   and $I'_i$
is a one-point interval  for every $i\in O\cap \{1,\dots, l\}$.

In particular, $F_l(b')$ holds. Now, by inspecting  the definition
of $F_l$ according to Equation \eqref{eq:tr-2} on page
\pageref{eq:tr-2}, it is easy to construct the required interval and
its partition. In all four cases $I_i$ is defined as $I'_i$ for $i<
l$ and we explain how $I_l$ and $I_{l+1}$ are defined.

If     $l\in O$ and $l+1\in O$, then    $F_l:=D_l\wedge \false
\until F_{l+1}$. Note that $F_l$, holds at $b'$, therefore  $b'$ has
a successor $c $ and $c\leq a_{l+1}$ because $b'\leq a_l<a_{l+1}$.
Define $I_l:=I'_l$, $b:=c$ and $I_{l+1}:=\{b\}$. It is clear that
$I_1, \dots ,I_{l+1}$ is a required partition.

If     $l\in O$ and $l+1\notin O$,   then $F_l:= D_l \wedge D_{l+1}
\until F_{l+1}$; hence,  there is $c>b'$ such that $F_{l+1}(c)$ and
$\delta_{l+1}$ holds on $(b',c]$.
%
%
%
Define $I_l:=I'_l$.
 Define $b:=\min(c,a_{l+1})$,
  and  $I_{l+1}$ as $(b',b]$.
 It is clear that   $I_1, \dots ,I_{l+1}$ is a required partition.

If $l\notin O$ and $l+1\in O$, then  $F_l:= D_l \wedge
 D_{l} \until F_{l+1}$; hence,
  there is $c>b'$ such that $F_{l+1}(c)$ and $\delta_{l}$
holds on $[b',c)$.
 Define $b:=\min(c,a_{l+1})$.
Define $I_{l}$ as $I'_l\cup [b',b)$ and $I_{l+1} $ as $\{b\}$.
%

If $l\notin O$ and $l+1\notin O$, then   $F_l:= D_l \wedge
 D_{l} \until^* F_{l+1}$. Since $F_l$ holds at $b'$ there is $c>b'$ and a partition of $[b',c]$ into two non-empty
 intervals $J_1$ and $J_2$ such that $J_1<J_2$ and
$D_l$ holds at all points of $J_1$ and $F_{l+1}$ holds at all points
of $J_2$. If $c<a_{l+1}$ define $I_l:=I'_l\cup J_1$ and
$I_{l+1}:=J_2$ and $b:=c$. If $a_{l+1}\in J_2$ define $I_l:=I'_l\cup
J_1$, $I_{l+1}:=J_2\cap\{a\mid a\leq a_{l+1}\}$ and $b:=a_{l+1}$. If
$a_{l+1}\in J_1$, define  $I_l:=I'_l\cup (J_1 \cap  \{a\mid a<
a_{l+1}\}) $, $I_{l+1}:=\{a_{l+1}\} $ and $b:=a_{l+1}$. It is clear
that  $b\leq a_{l+1}$ and $I_1, \dots ,I_{l+1}$ is a required
partition.
%
\end{proof}

\subsection{Proof of Proposition
\ref{lem:occ-st}}\label{sect:proof-lem:occ-st}
 The proof of  Proposition  \ref{lem:occ-st} is very similar to the proof of
Proposition \ref{lem:snegopen}. Its  Corollary \ref{cor:contS} will
be used in the proof of Proposition \ref{lem:snegopen}.

%
Let $A_n(P_1, \dots, P_n,z_0,z_1)$  be $  \E x_1 \dots  \E x_n
                             \left(z_0<x_1 <  \dots
  < x_n<z_1 \right)     \wedge ~\bigwedge_{i=1}^{n}  P_i(x_i)$.
  We have to prove that $\neg A_n$ is  equivalent to a simple formula.

  $\neg A_n$  is equivalent to the disjunction
of $(z_0,z_1)=\emptyset$ and of $(z_0,z_1) \NEQ      \emptyset
\wedge \neg A_n$. The first disjunct is equivalent to a simple
formula. Therefore, it is sufficient to prove that the second
disjunct is equivalent to a simple formula.

 Below we  assume that $(z_0,z_1)$ is non-empty, and prove by induction on $n$.

\emph{Basis}: The case $n=1$ is trivial.

\emph{Inductive step}: $n\mapsto n+1$.

Since $(z_0,z_1)$ is non-empty, then  one of the following cases
holds:
\begin{description}
\item[Case 1] There is no occurrence of $P_1$ in $(z_0,z_1)$ or there is no occurrence of $P_{n+1}$ in $(z_0,z_1)$.

\item[Case 2] $z_0=\inf\{z\in (z_0,z_1)\mid P_1(z)\}$.
\item[Case 2$'$] $z_1=\sup\{z\in (z_0,z_1)\mid P_{n+1}(z) \}$.
   This case is dual to   case 2.

\item[Case 3]  $\inf\{z\in (z_0,z_1)\mid P_1(z)\}$ is  an element in
  $(z_0,z_1)$.

  \item[Case 3$'$] $\sup\{z\in (z_0,z_1)\mid P_{n+1}(z) \}$ is  an element in
  $(z_0,z_1)$. This case is dual to   case 3.

\item[Case 4]
\begin{enumerate}
  \item  Both  $c:=\inf\{z\in (z_0,z_1)\mid P_1(z)\}$  and
  $d:=\sup\{z\in (z_0,z_1)\mid P_{n+1}(z) \}$  are gaps in
  $(z_0,z_1)$
    and
    \item $c\geq d$.
\end{enumerate}
\item[Case 5]
  \begin{enumerate}
  \item Both  $c:=\inf\{z\in (z_0,z_1)\mid P_1(z)\}$  and
  $d:=\sup\{z\in (z_0,z_1)\mid P_{n+1}(z) \}$  are gaps in
  $(z_0,z_1)$  and
\item   $c<d$.
  \end{enumerate}
\end{description}
For each of these    cases we construct a  simple formula $\cond_i$
which describes it (i.e., Case $i$ holds iff $\cond_i$ holds), and
show that if $\cond_i$ holds, then $\neg A_{n+1}$ is equivalent to a
simple formula $\form_i$.
 Hence, $\neg A_{n+1}$
 is equivalent to a simple formula
$\vee_i [ \cond_i\wedge\form_i]$.


\medskip \noindent \textbf{Case 1}
 This case holds iff
 $\Part{\neg P_1(x)}{\emptyset}(z_0,z_1) \vee \Part{\neg P_{n+1}(x)}{\emptyset}(z_0,z_1)$
 In this case $\neg A_{n+1}$   is equivalent to $\true$.

 \medskip \noindent \textbf{Case 2}
Case 2 holds iff  $\KPLUS(P_1)(z_0)$.
In  this case $\neg A_{n+1}$
 iff  $\neg A_n(P_2,\dots ,P_{n+1},z_0,z_1)$ which is equivalent to a simple  formula by the inductive assumption.


\medskip \noindent \textbf{Case 2$'$} This case is dual to Case 2.

\medskip \noindent \textbf{Case 3}
This case holds iff there is (a unique)  $r_0\in (z_0,z_1)$  such
that $\neg P_1$ holds along $(z_0,r_0)$ and either $P_1(r_0)$ or
$\KPLUS(P_1)(r_0)$.

This $r_0$ is definable by the following  simple  formula, i.e.,
$r_0$ is a unique $z$ which satisfies it:
\begin{align} 
\INF(P_1,z_0, z,z_1):= z_0<& z <z_1 \wedge ``\mbox{no $P_1$ in
$(z_0,z)$}"  \wedge  \notag \\ & \wedge ( P_1(z) \vee \KPLUS(
P_1)(z))\notag
\end{align}
Hence, this case is described by  $\Elub{ z}{z_0}{z_1} \INF(P_1,z_0,
z,z_1)$ which is equivalent to a simple formula.

  In this  case $\neg A_{n+1}$  iff
  $\Elub{z} {z_0}{z_1}\big(\INF(P_1,z_0,z,z_1) \wedge \neg A_n(P_2,\dots
,P_n,z,z_1)\big)$. The inductive assumption  and Lemma
\ref{lem:closure1} imply that this  formula is equivalent to a
simple formula.

\medskip \noindent \textbf{Case 3$'$} This case is dual to Case 3.

\medskip \noindent \textbf{Case 4}
The first condition holds iff
\begin{itemize}
\item
  $z_0$ succeeded by $\neg P_1$ gap in $(z_0,z_1)$, i.e.
 $\bgamma^+(\neg P_1)(z_0) $ and $P_1$  holds at some point
in $(z_0,z_1)$,
 and
 \item $z_1$ preceded by $\neg
 P_{n+1}$ gap in $(z_0,z_1)$, i.e.,   $\bgamma^-(\neg P_{n+1})(z_1) $  and $P_{n+1}$ holds at some point
in $(z_0,z_1)$.
\end{itemize}
(Modalities $\bgamma^+$ and $\bgamma^-$ were defined in Sect.
\ref{subsect:TL}.) Hence, the first condition is equivalent to a
simple formula.

 If the
first condition holds, then the second condition holds iff in
$(z_0,z_1)$ no occurrence of $P_1$ precedes an occurrence of
$P_{n+1} $, i.e., iff ${ \texttt{ Part} (\langle    \neg P_1,\neg
P_{n+1}   \rangle ,\emptyset )}(z_0,z_1)$.
Hence, Case 4 is described by a simple formula.

In Case 4  $\neg A_{n+1}(P_1, \dots, P_{n+1},z_0,z_1)$ is equivalent
to $\true$.

\medskip \noindent \textbf{Case 5}
The first condition  is the same  as in Case 4.
If the first condition holds, then   $z$  is between $c$ and $d$ iff
$z$ satisfies the formula:
$$\Beet (z_0,z,z_1):=  \Elub{x_1}{z_0}{z}P_1(x_1)\wedge
\Elub{x_{n+1}}{z}{z_1}P_{n+1}(x_{n+1}).$$
Hence, this case can be described as the conjunction of the first
condition and $\E z \Beet (z_0,z,z_1)$ and this is equivalent to a
simple formula.

Note that in  this case $\E x_1 \dots  \E x_{n+1}
                             \left(z_0<x_1 <  \dots
  < x_{n+1}<z_1 \right)     \wedge ~\bigwedge_{j=1}^{n+1}  P_j(x_j)
  $ holds iff for every $z$ between  $c$ and $d$ one of the
  following  $2n-1$ conditions holds:
for $i=1,\dots ,n$:
\begin{gather*}
 \E x_1\dots  \E x_{n+1}
                             \left(z_0<x_1 <  \dots
  < x_{n+1}<z_1 \right)  \wedge x_i<z<x_{i+1}   \wedge ~\bigwedge_{j=1}^{n+1}
  P_j(x_j)\\
 \intertext{for $i=2,\dots ,n$:  }
\E x_1\dots  \E x_{n+1}
                             \left(z_0<x_1 <  \dots
  < x_{n+1}<z_1 \right)  \wedge x_i=z   \wedge ~\bigwedge_{j=1}^{n+1}
  P_j(x_j)
\end{gather*}
%
Hence, $\neg \E x_1 \dots  \E x_{n+1}
                             \left(z_0<x_1 <  \dots
  < x_{n+1}<z_1 \right)     \wedge ~\bigwedge_{j=1}^{n+1}  P_j(x_j)$
is equivalent to
\begin{align*}
        \E z (&\Beet  (z) \wedge
                             \bigwedge _{k=1}^{n}  \big[ \neg A_{k}(P_1, \dots, P_{k},z_0,z
                             ) \vee \neg A_{n+1-k }(P_{k+1 }, \dots, P_{n+1},z,z_1)
                             \big]
                             \\
                      \wedge  &      \bigwedge _{k=2}^{n}  \big[ \neg A_{k-1}(P_1, \dots, P_{k-1},z_0,z
                             )\vee \neg P_k(z)  \vee \neg A_{n+1-k }(P_{k+1 }, \dots, P_{n+1},z,z_1)
                             \big]
\end{align*}
%
%
By the inductive assumption  $ \neg A_{k}$ and  $\neg A_{n+1-k } $
are
   simple for $k=1,\dots, n$. Since $ \Beet$  is
 a  simple formula, and the set of simple formulas is closed under
conjunction, disjunction and   existential quantifier, we obtain   a
formalization of this case by a   simple  formula. This completes
the proof of Proposition \ref{lem:occ-st}.

By Proposition \ref{lem:occ-st}, Lemma \ref{lem:st-form3} and
standard logical equivalences we derive:
\begin{cor} \label{cor:contS}
\leavevmode
\begin{enumerate}
\item $\neg\Elub{z}{z_0}{z_1} \Part{ \delta'_1, \dots , \delta'_n}{O'} (z_0,z] $
%
 is
 equivalent   to a simple formula.

\item
$\neg\Elub{z}{z_0}{z_1} \Part{ \delta'_1, \dots , \delta'_n}{O'}
[z,z_1) $
 is
 equivalent   to a simple  formula.
\end{enumerate}
\end{cor}

\begin{proof}
\leavevmode
\begin{enumerate}
\item
 Set $k:=n+1$, $\delta_1:=\true$,
  $\delta_{i+1}:=\delta'_i$ for $i=1,\dots ,n$ and
$O:= \{1\}\cup \{i+1\mid i\in O'\}$. Observe:  $\Part{ \delta'_1,
\dots , \delta'_n}{O'} (z_0,z]$ iff $\Part{\delta_1, \dots ,
\delta_k}{O}[z_0,z]$.

Let $F_i$ be defined as in Lemma \ref{lem:st-form3}.  Then $ \E x_2
\dots  \E x_{k-1} z_0<x_2< \dots <x_{k-1} <x_k \wedge F_1(z_0)\wedge
\bigwedge_{i=2}^k F_i(x_i)$ iff $\exists z  (z_0<z \leq x_k \wedge
\Part{\delta_1, \dots , \delta_k}{O}[z_0,z])$.

Hence, $\neg \Elub{z}{z_0}{z_1} \Part{ \delta'_1, \dots ,
\delta'_n}{O'} (z_0,z] $ is equivalent to $\neg F_1(z_0) \vee \neg \
\E x_2 \dots  \E x_{k} z_0<x_2< \dots <x_{k-1} <x_k<z_1 \wedge
\bigwedge_{i=2}^k F_i(x_i)$. The first disjunct is an atom (in the
canonical expansion) and the second disjunct is equivalent to a
simple formula by Proposition \ref{lem:occ-st}. Therefore, $\neg
\Elub{z}{z_0}{z_1} \Part{ \delta'_1, \dots , \delta'_n}{O'} (z_0,z]
$ is equivalent to
 a simple formula.

 \item is the
mirror image of (1).
 \qedhere
\end{enumerate}
\end{proof}

\subsection{Proof of Proposition  \ref{lem:snegopen}}\hfill\\
\textbf{Convention.} We often will say ``a formula is simple"
instead of ``a formula is equivalent to a simple formula." In all
such  cases equivalence to a simple formula is proved by Lemma
\ref{lem:closure1} and by  standard logical transformations and/or
using the inductive hypotheses.

\label{sect:proof-snegopen}
We proceed by induction on $n$.

\emph{Basis.} The    case $n=1$ is immediate.

\emph{Inductive step} $n\mapsto n+1$.

$\neg \Parts{1}{n+1}{O} (z_0,z_1)$  is equivalent to the disjunction
of $(z_0,z_1)=\emptyset$ and of $(z_0,z_1) \NEQ      \emptyset
\wedge \neg \Parts{1}{n+1}{O} (z_0,z_1)$. The first disjunct is
equivalent to a simple formula. Therefore, it is sufficient to prove
that the second disjunct is equivalent to a simple formula.

From now on we assume that  $(z_0,z_1)$ is non-empty.

Observe that one of the following cases holds:
\begin{description}
%

\item[Case 1] $\delta_1$ holds on all points in  $(z_0,z_1)$.

\item[Case 1$'$] $\delta_{n+1}$ holds on all points in  $(z_0,z_1)$.
 This case is dual to   case 1.

\item[Case 2]
$z_0=\inf\{z\in (z_0,z_1)\mid \neg \delta_1(z)\}$ or $z_1=\sup\{z\in
(z_0,z_1)\mid \neg \delta_{n+1}(z)\}$.

\item[Case 3]

  $\inf\{z\in (z_0,z_1)\mid\neg  \delta_1(z)\}$ is  an element in
  $(z_0,z_1)$.

  \item[Case 3$'$]
    $\sup\{z\in (z_0,z_1)\mid \neg \delta_{n+1}(z) \}$ is  an element in
  $(z_0,z_1)$. This case is dual to  case 3.

\item[Case 4]
  Both  $c:=\inf\{z\in (z_0,z_1)\mid \neg \delta_1(z)\}$  and
  $d:=\sup\{z\in (z_0,z_1)\mid \neg \delta_{n+1}(z) \}$  are gaps in
  $(z_0,z_1)$
    and
     $c> d$.

\item[Case 5]
    Both  $c:=\inf\{z\in (z_0,z_1)\mid  \neg\delta_1(z)\}$  and
  $d:=\sup\{z\in (z_0,z_1)\mid \neg \delta_{n+1}(z) \}$  are gaps in
  $(z_0,z_1)$  and
    $c<d$.
\item[Case 6]
  Both  $c:=\inf\{z\in (z_0,z_1)\mid \neg \delta_1(z)\}$  and
  $d:=\sup\{z\in (z_0,z_1)\mid \neg \delta_{n+1}(z) \}$  are gaps in
  $(z_0,z_1)$
    and
     $c= d$.

\end{description}
For each of these    cases we construct a  simple formula $\cond_i$
which describes it (i.e., Case $i$ holds iff $\cond_i$ holds), and
show that if $\cond_i$ holds, then $\neg \Parts{1}{n+1}{O}
(z_0,z_1)$ is equivalent to  a simple formula $\form_i$.
 Hence, $\neg
\Parts{1}{n+1}{O} (z_0,z_1)$ is equivalent to a simple formula
$\vee_i [ \cond_i\wedge\form_i]$.


\medskip \noindent \textbf{{Case 1}}
is described by ${ \texttt{Part} (\langle   \delta_1
\rangle ,\emptyset )}(z_0,z_1)$.
%
%
%
 In this case
 $\neg \Parts{1}{n+1}{O} (z_0,z_1)$ is equivalent to
$\neg\Elub{z}{z_0}{z_1}
 \Parts{1}{n+1}{O} [z,z_1)$, and by
Corollary  \ref{cor:contS}  this is a simple  formula.

\medskip \noindent \textbf{Case 1$'$}
  This case is dual to  Case 1.

\medskip \noindent
 \textbf{{Case 2}}\sloppy{
 This case is described by
$  \KPLUS(\neg \delta_1)(z_0) \vee \KMINUS(\neg \delta_{n+1})(z_1)$.
In this case $\neg \Parts{1}{n+1}{O} (z_0,z_1)$ is equivalent to
$\true$.}

Note that in  the above   cases we have not used the  inductive
assumption. Case  6 will be also proved directly. However, in cases
3-5  we will use the inductive assumption.

We  introduce notations and state an observation which will be used
several times.

For a set $O$ of natural numbers  and $i\in \nat$,  we denote by
$\sh{i}$ the set $O$ shifted by $i$, i.e.,  $\sh{i}:=\{j \mid j>0
\wedge j+i\in O\}$.

Define
\begin{align*} C^{<i}(z_0,z):= & \begin{cases}
 \mbox{``$z$ is the successor of $z_0$''} & \mbox{for }i=1\\
\Parts{1}{i-1}{O\cap \{1,\dots, i-1\}}
(z_0,z) & \mbox{for $i=2,\dots ,n+2$} \\
\end{cases}\\
  C^{>i}(z,z_1):=&\begin{cases}
 \mbox{``$z_1$ is the successor of $z$"}  & \mbox{for }i=n+1\\
\Parts{i+1}{n+1}{O_{sh(i)} }
(z,z_1) & \mbox{for $i=0,\dots ,n$}\\
\end{cases}
\end{align*}
 For $i=1,\dots , n+1$ define
\begin{align*}
%
C^{\leq i}(z_0,z):= &  C^{<i}(z_0,z) \vee C^{<i+1}(z_0,z )\\
%
%
%
%
 C^{\geq i}(z,z_1):= &  C^{>i}(z,z_1) \vee C^{>i-1}(z,z_1)\\
 %
A_i(z_0,z,z_1):= &\begin{cases}
 C^{<i}(z_0,z)\wedge  \delta_{i}(z) \wedge  C^{>i}(z,z_1) & \mbox{if
 } i\in O\\
 C^{\leq i}(z_0,z)\wedge  \delta_{i}(z) \wedge
C^{\geq i}(z,z_1) & \mbox{otherwise}\\
\end{cases}
\end{align*}
%
From these definitions we obtain the following equivalences:
\begin{gather}\label{eq:eqA}
 \Parts{1}{n+1}{O} (z_0,z_1)
             \Leftrightarrow   \Elub{z}{z_0}{z_1}A_i  \qquad \mbox {for } i\in 1,\dots ,n+1\\
 \intertext{and   if   $(z_0,z_1) \NEQ      \emptyset  $, then  }
  \Parts{1}{n+1}{O} (z_0,z_1)  \Leftrightarrow \Alub{z}{z_0}{z_1} \big( \bigvee_{i }
  A_i\big) \label{eq:eqC}
\end{gather}
Since, we assumed that $(z_0,z_1)$ is non-empty,  by
\eqref{eq:eqA}-\eqref{eq:eqC} we have 
\begin{gather*}
\neg \Parts{1}{n+1}{O} (z_0,z_1) \intertext{ is equivalent  to}
               \Elub{z}{z_0}{z_1} \big( \bigwedge_{i } \neg A_i \big) \intertext{and to } \Alub{z}{z_0}{z_1}  \big(
\bigwedge_{i } \neg A_i\big)
\end{gather*}

Hence, for every $\varphi(z_0,z,z_1)$
\begin{gather*}\Elub{z}{z_0}{z_1}
\varphi(z)\wedge \neg \Parts{1}{n+1}{O} (z_0,z_1) \intertext{ is
equivalent to}
  \Elub{z}{z_0}{z_1}\Big( \varphi(z) \wedge \big( \bigwedge_{i } \neg
  A_i
 \big)\Big) \intertext{ is equivalent to}
\Elub{z}{z_0}{z_1} \Big( \big( \varphi(z)\wedge \bigwedge_{i
\in\{2,\dots,n\} } \neg A_i \big) \wedge \big( \varphi(z)  \wedge
\neg A_1\wedge\neg  A_{n+1} \big)\Big)
 \end{gather*}
 By the inductive assumption, the definition of $A_i$,  and Lemma \ref{lem:closure1},
  we obtain that $\neg A_i$ are simple formulas
   for $i \in\{2,\dots,n\}$.  Similarly, if $1 \in O$ (respectively, $n+1\in O$), then
   $\neg A_1$ (respectively, $\neg A_{n+1}$) is equivalent to a simple formula.  The set
of simple formulas is closed under $\wedge$, $\vee$ and $\exists$.
Hence,

\begin{obs} \label{obs}
Assume that  $\vp(z)$ is equivalent to a simple formula, and if $1
\notin O$, then  $\varphi(z) \wedge \neg A_1$ is equivalent to a
simple formula,  and   if $n+1 \notin O$, then $\varphi(z)
\wedge\neg A_{n+1}$ is  equivalent to a simple formula. Then
$\Elub{z}{z_0}{z_1} \varphi(z)\wedge \neg \Parts{1}{n+1}{O}
(z_0,z_1)$
  is
equivalent to a simple   formula.
\end{obs}
%
In cases 3-5 we will use this observation with some instances of
$\varphi$.
%

\medskip \noindent \textbf{{Case 3}}
This case holds iff there is (a unique)  $r_0\in (z_0,z_1)$  such
that $\delta_1$ holds along $(z_0,r_0)$ and $\neg\delta_1(r_0)\vee
 \KPLUS(\neg \delta_1)(r_0)$.

This $r_0$ is definable by the following  simple  formula, i.e.,
$r_0$ is a unique $z$ which satisfies it:
\begin{align}   
\INF_{\neg\delta_1}(z_0, z,z_1):= z_0<& z <z_1 \wedge
(\suc(z_0,z)\vee { \texttt{ Part} (\langle   \delta_1 \rangle
,\emptyset )}(z_0,z)) \wedge  \notag \\ & \wedge ( \neg \delta_1(z)
\vee \KPLUS(\neg \delta_1)(z)) \notag
\end{align}
Hence, this case is described by a simple formula $\Elub{
z}{z_0}{z_1} \INF_{\neg\delta_1}(z_0, z,z_1)$.

 By Observation \ref{obs}
it is sufficient to prove  that (1)  if $1\notin O$ then
$\INF_{\neg\delta_1}    \wedge \neg A_1$  is equivalent to a simple
formula, and (2) if $n+1\notin O$, then
  $\INF_{\neg\delta_1}   \wedge  \neg A_{n+1}$ is   equivalent to a simple  formula.



Note that $  \neg \delta_1(z) \vee \KPLUS(\neg \delta_1)(z) $
implies $\neg \big( \delta_1(z)  \wedge \Parts{1}{n+1}{O}
(z,z_1)\big)$. Therefore, by the definition of $A_1$ for the case
when $1\notin O$, and standard logical transformations we obtain
that $ \INF_{\neg\delta_1}
  \wedge \neg A_1$ is equivalent to  $ \INF_{\neg\delta_1} \wedge
  \big( \neg C^{\leq 1} \vee
\neg   \delta_1(z)\vee \neg \Parts{2}{n+1}{O_{sh(1)}} (z,z_1)$. The
last formula is equivalent to a simple formula by the inductive
assumption and standard logical equivalences.

%

%
 If
  $n+1\notin O$, then
$\INF_{\neg\delta_1}   \wedge \neg A_{n+1}$ is
 equivalent to $$\INF_{\neg\delta_1} (z_0, z,z_1)  \wedge (\neg  C^{\geq n+1}(z,z_1)\vee \neg \delta_{n+1}(z)\vee
 \neg C^{\leq n+1}(z_0,z)).
$$
$\neg  C^{\geq n+1}(z,z_1)$ is a simple formula by the induction
basis. Note that $\INF_{\neg\delta_1} (z_0, z,z_1)$ implies
$suc(z_0,z)$ or ``$\delta_1$ holds along the interval $(z_0,z)$." By
Case 1 the conjunction of ``$\delta_1$ holds along the interval
$(z_0,z)$" and $\neg C^{\leq n+1}(z_0,z))$ is a simple formula.
Therefore, $\INF_{\neg\delta_1} \wedge \neg A_{n+1}$ is equivalent
to
 a simple  formula.

 \medskip \noindent \textbf{{Case 3$'$}} This case is dual to case 3.


 \medskip \noindent \textbf{Case 4}
The conjunction of the following conditions expresses by a simple
formula that $z$ is in the interval $(d,c)$:
\begin{itemize}
\item $z_0$ succeeded by $\delta_1$ gap in $(z_0,z_1)$ -
$\bgamma^+(\delta_1)(z_0)$ and $\neg \delta_1$ holds at some point
in $(z_0,z_1)$.
\item $z_1$ preceded by $\delta_{n+1}$ gap in $(z_0,z_1)$ -
$\bgamma^-(\delta_{n+1})(z_1)$ and $\neg \delta_{n+1}$ holds at some
point in $(z_0,z_1)$.
\item $\delta_1$ holds along $(z_0,z)$ and  $\delta_{n+1}$ holds along
$(z,z_1)$.
\end{itemize}
Let us denote this conjunction by $\In1 (z_0,z,z_1)$.

Hence, this case  holds iff $\Elub{z}{z_0}{z_1} \In1(z_0,z,z_1)$.

By Observation \ref{obs}  it is sufficient to show that (1)  if
$1\notin O$, then $\In1(z_0,z,z_1)    \wedge \neg A_1$  is
equivalent to a simple formula, and (2) if $n+1\notin O$, then
 $ \In1(z_0,z,z_1)   \wedge  \neg A_{n+1}$ is   equivalent to a simple  formula.

if $1\notin O$ then $\In1(z_0,z,z_1)\wedge \neg A_1(z_0,z,z_1)$ is
equivalent to
\[\In1(z_0,z,z_1)\wedge \big(
\neg \delta_1(z)\vee \neg  C^{\leq 1} (z_0,z) \vee   (\neg  C^{>1}
(z,z_1)\wedge \neg  C^{>0}(z,z_1))\big).\]
$\In1(z_0,z,z_1)$ implies that $\delta_{n+1}$ holds along $(z,z_1)$,
therefore, by Case 1$'$  both $\In1 \wedge  \neg C^{>0}(z,z_1)$ and
$\In1\wedge  \neg C^{>1}(z,z_1)$ are simple. By the basis of
induction
 $\neg  C^{\leq 1}$
is   simple. Hence, $\In1(z_0,z,z_1)\wedge \neg A_1(z_0,z,z_1)$ is
simple.

Similar   arguments show that if $n+1\notin O$, then
$\In1(z_0,z,z_1)\wedge \neg A_{n+1}(z_0,z,z_1)$  is   simple.

 \medskip \noindent \textbf{Case 5}
It is easy to write a simple  formula $\btw(z_0, z,z_1)$ which
expresses that $z$ is in the interval  $(c, d)$.
 $\btw(z_0, z,z_1)$ can be defined as the conjunction of $z_0<z<z_1$ and of
 \begin{itemize}
\item
$z_0$ succeeded by $\delta_1$ gap in $(z_0,z )$ -
$\bgamma^+(\delta_1)(z_0)$ and $\neg \delta_1$ holds at some point
in $(z_0,z )$.

\item $z_1$ preceded by $\delta_{n+1}$ gap in $(z,z_1)$ -
$\bgamma^-(\delta_{n+1})(z_1)$ and $\neg \delta_{n+1}$ holds at some
point in $(z,z_1)$.
\end{itemize}
%
 Hence, this case  holds iff
$\Elub{z}{z_0}{z_1} \btw(z_0,z,z_1)$.

By Observation \ref{obs}  it is sufficient to show that (1)  if
$1\notin O$, then $\btw(z_0,z,z_1)\wedge \neg A_1(z_0,z,z_1)$ is
equivalent to a simple formula,  and (2) if $n+1\notin O$, then
$\btw(z_0,z,z_1)\wedge \neg A_{n+1}(z_0,z,z_1)$ is equivalent to a
simple formula.
 Since  $\btw$ 
implies $\neg  C^{\leq 1}$ it follows that
 $\btw\wedge \neg A_1$ is equivalent to
$\btw$. Since   $\btw$
implies $\neg C^{\geq n+1}$ it follows that $\btw\wedge \neg
A_{n+1}$ is equivalent to $\btw$.  Therefore, both $\btw\wedge \neg
A_1$ and $\btw\wedge \neg A_{n+1}$ are simple. 

\medskip \noindent\textbf{Case 6}
 Both  $c:=\inf\{z\in (z_0,z_1)\mid\neg \delta_1(z)\}$  and
  $d:=\sup\{z\in (z_0,z_1)\mid \neg \delta_{n+1}(z) \}$  are gaps in
  $(z_0,z_1)$
    and
     $c\geq d$ iff
 the conjunction of the following   holds:
\begin{enumerate}
\item $z_0$ succeeded by $\delta_1$ gap in $(z_0,z_1)$.
\item $z_1$ preceded by $\delta_{n+1}$ gap in $(z_0,z_1)$.
\item ${ \texttt{ Part} (\langle   \delta_1,\delta_{n+1}
\rangle ,\emptyset )}(z_0,z_1)$.
\end{enumerate}
%
If (1)-(3) holds, then $ d<c$ iff $F(z_0)$ defined as
$\delta_1\until (\delta_1 \wedge \untg(\delta_1,\delta_2))(z_0)$
holds,
%
%
where $\untg$ is defined on page
 \pageref{lem:trans-form2}.
%

Hence, this case can be described by the conjunction of (1)-(3) 
 and
 $\neg F(z_0)$. (1) and (2) are expressed by simple formulas like in Case
 4;
 (3) and
 $\neg F(z_0)$ are simple formulas. Therefore, this case is
 described by a simple formula.

In this case $\Parts{1}{n+1}{O} (z_0,z_1)$ holds iff there is $i$
such that
 $\Part{
\delta_1, \dots , \delta_i}{O}$ holds on $ (z_0,c)$ and  $\Part{
\delta_{i}, \dots , \delta_{n+1}}{O}$ or
 $\Part{
\delta_{i+1}, \dots , \delta_{n+1}}{O} $ holds on $(c,z_1)$.
Applying  Lemma  \ref{lem:trans-form2} to the tuple  $ \langle
\true, \delta_1,\dots ,\delta_i \rangle$,
  $O:=\{1\}\cup \{j+1 \mid j\in O\wedge j\leq i\} $  and $\delta_1$, we
  obtain a temporal formula
$F_i$ such that $F_i(z_0)$ iff  $\Part{ \delta_1, \dots ,
\delta_i}{O} $ holds on $(z_0,c)$. By  the mirror  arguments there
is a
temporal formula $H_i$ such that $H_i(z_1)$ iff $ \Part{ \delta_{i},
\dots , \delta_{n+1}}{O}$ holds on  $(c,z_1) $. Hence, in this case
$\neg \Parts{1}{n+1}{O} (z_0,z_1)$ is equivalent to
$$\bigwedge_{i=1}^n \big(\neg F_i(z_0) \vee ( \neg H_i(z_1) \wedge
 \neg H_{i+1}(z_1) )\big).$$

\section{Related Works} \label{sect:related}
  Our proof is very similar to the proof of Kamp's theorem in
\cite{Rab14}. The only novelty of  our proof are  partition
formulas. Simple  partition formulas generalize $\EA$-formulas which
played a major role in the proof of  Kamp's theorem \cite{Rab14}.
Roughly speaking an $\EA$-formula is a normal partition formula
which uses only basic partition expressions $\Parts{1}{n}{O} $ with
the following restriction:  for $i<n$, if $i\notin O$ then $i+1\in
O$. This restriction implies that if a partition $I_1,\dots, I_n$
witnesses that an interval $[t,t']$ of $\mM$ satisfies
$\Parts{1}{n}{O} $, then all intervals $I_i$ have endpoints in
$\mM$. Over the Dedekind complete orders all intervals have
end-points and every partition expression is equivalent to a
disjunction of the restricted partition expressions; however, over
general linear orders $\Part{P_1(x), P_2(x)
 }{\emptyset}$  is not equivalent to a positive boolean combination
 of restricted partition expressions.

As far as we know, there are only two published proofs of Stavi's
theorem. One  is based on   separation  in  Chapter 11 of
\cite{GHR94}, and the other is   based on games in  \cite{GHR93}
(reproduced in Chapter 12 of  \cite{GHR94}). They are   much  more
complicated than the proofs of Kamp's theorem in \cite{GHR94}.

 A temporal logic has the
\emph{separation} property if its formulas can be equivalently
rewritten as a boolean combination of formulas, each of which
depends only on the past, present  or future. The separation
property was introduced by  Gabbay \cite{Gab81}, and
 surprisingly,  a temporal logic which can express  $\G $ and  $\Gb $ has the separation property (over a  class $\mC$ of structures) iff it is
expressively  complete for \FOMLOiB  over $\mC$.

 In the proof based
on separation, a special temporal  language $L*$ is carefully
designed. The formulas of $L*$ are evaluated over Dedekind-complete
chains. For every chain $\mM$ its completion $\mM^c$ is defined. It
is shown: (1)  $L*$ has the separation property over the completions
of chains; (2) for every $\vp\in L*$ there is a formula $\psi\in
\TLs$ such that $\mM, t\models \psi$ iff $\mM^c,t\models \vp$, and
(3) for every formula $\xi(x)\in \FOMLOiB$ there is $\vp\in L*$ such
that $\mM, t\models \xi$ iff $\mM^c,t\models \vp$.

In the game-based proof for every  chain $\mM$ and $r\in \nat$  a
chain $\mM_r$  is defined.
 $\mM_r$   is the completion of $\mM$ by the gaps definable
by $\TLs$ formulas of the nesting depth $r$. Then,  special games on
the temporal structures are considered. The game arguments are
easier to grasp, then the separation ones, but they use complicated
inductive assertions.

%
Our proof    
 avoids completions and games  and  separates
general logical equivalences and temporal arguments. The proof is
similar to our proof of Kamp's theorem \cite{Rab14}; yet it is
longer because it treats some additional cases related to gaps in
time flows.
\section*{Acknowledgement}
I would like to thank an anonymous referee for  insightful
suggestions.


\end{document}